\documentclass{amsart}

\usepackage{amssymb}
\usepackage{amsmath}
\usepackage{amsthm}
\usepackage{amsbsy}
\usepackage{bm}

\usepackage{hyperref}
\usepackage{url}
\usepackage{dblfloatfix}
\usepackage{mdframed}

\usepackage{listings}
\usepackage{color}
\usepackage{enumitem}
\setlistdepth{7}
\renewlist{itemize}{itemize}{7}
\setlist[itemize,1]{label=\textbullet}
\setlist[itemize,2]{label=--}
\setlist[itemize,3]{label=*}
\setlist[itemize,4]{label=-}
\setlist[itemize,5]{label=$\circ$}
\setlist[itemize,6]{label=-}

\theoremstyle{plain}
\newtheorem{theorem}{Theorem}[section]
\newtheorem{lemma}[theorem]{Lemma}

\newtheorem{question}[theorem]{Question}
\newtheorem{conjecture}[theorem]{Conjecture}

\theoremstyle{definition}
\newtheorem{definition}[theorem]{Definition}

\theoremstyle{remark}
\newtheorem{remark}{Remark}[subsection]
\newtheorem*{acknowledgements}{Acknowledgements}

\newcommand{\Z}{\mathbb{Z}}
\newcommand{\R}{\mathbb{R}}

\newcommand{\N}{\mathbb{N}}

\newcommand{\F}{\langle \alpha, \beta \rangle}

\begin{document}

\title[Intertwined computer \& human proofs]{Homogeneous length functions on Groups: \\ Intertwined computer \& human proofs}

\author{Siddhartha Gadgil}

\address{	Department of Mathematics,\\
		Indian Institute of Science,\\
		Bangalore 560012, India}

\email{gadgil@iisc.ac.in}



\date{\today}

\begin{abstract}
We describe  a case of an interplay between human and computer proving which played a role in the discovery of an interesting mathematical result~\cite{polymath}. The unusual feature of the use of computers here was that a computer generated but human readable proof was read, understood, generalized and abstracted by mathematicians to obtain the key lemma in an interesting mathematical result.
\end{abstract}

\maketitle

\section{Introduction}

Computers have come to play many roles in mathematical proofs. Computer experimentation is commonly used to make conjectures and computer algebra systems are used for sophisticated calculations. Components of proofs of important results have also been provided by computers. Such rigorous computer proofs often generate independently verifiable \emph{proof certificates}. Proof assistants have been used to formalize proofs, including some very complex ones.

Here we describe a case different from these -- where a computer generated but \href{https://github.com/siddhartha-gadgil/Superficial/wiki/A-commutator-bound}{human readable} proof  was read, understood, generalized and abstracted by mathematicians to obtain the key lemma in a significant mathematical result. So far as we know this is the only such instance so far.

The result we discuss concerned a question about the existence of so called \emph{homogeneous length functions}, which was asked by Terrence Tao on his blog (Apoorva Khare had asked Tao this question). The question was answered in six days in a collaboration that became PolyMath~14\footnote{\textbf{Participants:} T. Fritz, S. Gadgil, A. Khare, P. Nielsen, L. Silberman, T. Tao.}, and the answer (and stronger results) have been published in~\cite{polymath}.

To state the main question, we need some definitions. We emphasise that in general the groups $G$ we consider are not abelian (commutative), i.e., if $x, y\in G$, we may have $xy\neq yx$. Thus the notation we use is multiplicative, similar to that for matrix multiplication (except with the identity denoted as $e$ rather than $I$). Recall that, for fixed $n\in\N$, \emph{invertible} $n\times n$ matrices form a group.

We sometimes denote the product of $x$ and $y$ as $x\cdot y$ instead of $xy$ for readability.

\begin{definition}
	A \emph{pseudo-length function} on a group $G$ is a function $l: G \to [0, \infty)$ such that
	\begin{itemize}
		\item $l(e) = 0$, where $e\in G$ is the identity.
		\item $l(g^{-1}) = l(g)$ for all $g \in G$ (\textbf{symmetry}).
		\item $l(gh) \leq l(g) + l(h)$ for all $g,h\in G$ (the \textbf{triangle inequality}).
	\end{itemize}
\end{definition}

\begin{definition}[Conjugacy invariance]
	A pseudo-length function $l$ on a group $G$ is said to be \emph{conjugacy invariant} if $l(ghg^{-1}) = l(h)$ for all $g, h\in G$.
\end{definition}

Recall that elements $x, y\in G$ are conjugate if there exists $g\in G$ such that $y = gxg^{-1}$. Conjugacy invariance is thus the property that conjugate elements have equal lengths. Note that in an abelian group, conjugate elements are equal, so this property is automatically satisfied.

\begin{definition}[Homogeneity]
	A pseudo-length function $l$ on a group $G$ is said to be \emph{homogeneous} if $l(g^n) = n\cdot l(g)$ for all $g\in G$, $n \in\Z$.
\end{definition}

\begin{definition}[Positivity]
	A pseudo-length function $l$ on a group $G$ is said to be a \emph{length function} if $l(g) > 0$ for all $g\in G \setminus \{ e \}$.
\end{definition}

If $G = (V, +)$ is the additive group of a vector space $V$ over $\R$, then a \emph{norm} on $V$ gives a homogeneous, conjugacy invariant length function. For example on $\R^2$ both $l_1(x, y) = |x| + |y|$ and $l_2(x, y) = \sqrt{x^2 + y^2}$ are homogeneous, conjugacy invariant length functions. To see this, note that the properties of a pseudo-length follow from the definition of norms. As mentioned above, conjugacy invariance is automatic as additive groups of vector spaces are abelian.

It was generalizing norms on Vector Spaces that motivated the main question  (we elaborate on this after stating Question~\ref{qn}). The question was formulated in terms of free groups as these are the prototypical non-abelian groups.

Recall that the free group $\F$ on $2$ generators $\alpha$ and $\beta$ is the group whose elements are \emph{equivalence classes} of words in $S =\{\alpha$, $\beta$, $\alpha^{-1}$, $\beta^{-1}\}$, where we think of $\alpha^{-1}$ and $\beta^{-1}$ as simply formal symbols (we will see that in $\F$ their equivalence classes are inverses of the equivalence classes of  $\alpha$ and $\beta$). Namely, we define an equivalence relation $\sim$ so that two words equivalent if and only if they are related by a  sequence of moves given by cancellation of pairs of \textbf{adjacent} letters that are \textbf{inverses}
of each other and its inverse move, namely \emph{inserting} a cancelling pair of letters. For example, in $\F$, $\alpha\beta\beta^{-1}\alpha\beta\alpha^{-1} = \alpha\alpha\beta\alpha^{-1}$
as cancelling the second and third letters of $\alpha\beta\beta^{-1}\alpha\beta\alpha^{-1}$ gives $\alpha\alpha\beta\alpha^{-1}$.
Conversely, inserting $\beta\beta^{-1}$ between the first and second letters of $\alpha\alpha\beta\alpha^{-1}$ gives $\alpha\beta\beta^{-1}\alpha\beta\alpha^{-1}$.

Formally, we consider the equivalence relation $\sim$ on words in $S$ generated by $$\xi_1\xi_2\dots\xi_m\lambda\lambda^{-1}\xi_{m+1}\dots \xi_n \sim \xi_1\xi_2\dots\xi_m\xi_{m+1}\dots \xi_n$$ where $\lambda\in S$, $\xi_i\in S \ \forall i, 1\leq i\leq n$ and $0\leq m\leq n$.
The case $m=0$ corresponds to prepending a cancelling pair, and $m=n$ to appending a cancelling pair. The case $n=0$ (which forces $m=0$) corresponds to the empty word.
The elements of $\F$ are equivalence classes under this equivalence relation.

Multiplication in $\F$ is given by concatenation, i.e.
$$(\xi_1\xi_2\dots \xi_n) \cdot (l'_1l'_2\dots l'_m) = \xi_1\xi_2\dots \xi_nl'_1l'_2\dots l'_m$$
where $\cdot$ denotes the group multiplication. More formally, concatenation of words induces a well-defined multiplication on equivalence classes under $\sim$ of words. The identity is the empty word $e$ (more formally the equivalence class of $e$), and the inverse of an element is obtained by inverting letters and reversing the order, i.e., $(\xi_1\xi_2\dots \xi_n)^{-1}=\xi_n^{-1}\dots \xi_2^{-1}\xi_1^{-1}$.

We can now state the main question that was studied.

\begin{question}\label{qn}
	Is there a homogeneous, conjugacy-invariant length function $l$ on the free group $\F$ on $2$ generators?
\end{question}

Khare asked this question motivated by wanting to generalize results of Khare-Rajaratnam~\cite{KR1}\cite{KR2} from vector spaces with norms to a more general context where commutativity was no longer assumed. However, it was not clear whether any (additional) examples would satisfy this more general hypothesis. The free group was taken as a prototypical group which is not abelian. In fact the results of~\cite{polymath} show that, in a strong sense, the only groups having homogeneous, conjugacy invariant length functions are abelian groups, and all such functions are restrictions of norms to subgroups of vector spaces.

The question is also natural from the point of view of Geometric group theory, where length functions are a central concept and conjugacy invariance of lengths (which corresponds to bi-invariance of metrics) is also commonly studied. Lengths satisfying the additional condition of \emph{homogeneity} were not much studied previously -- which we now know is because there are no interesting examples (except restrictions of norms on vector spaces, which are well understood).

\section{Homogeneous length functions and the Internal repetition trick}\label{S:Homogeneous}

We now describe the history and some ingredients of the solution Question~\ref{qn}.

It is natural to view Question~\ref{qn} as asking whether a homogeneous, conjugacy-invariant pseudo-length function $l$ on $\F$ can also be positive, hence a length function. Further, we can normalize to assume that $l(s)\leq 1$ for $s=\alpha, \beta$ (hence, by symmetry, $l(s)\leq 1$ for $s = \alpha^{-1}, \beta^{-1}$). We shall say $l$ is \emph{normalized} if $l(s)\leq 1$ for $s=\alpha, \beta, \alpha^{-1}, \beta^{-1}$.

After the failure of various constructions (by day $3$), the following conjecture seemed likely.

\begin{conjecture}
	For any homogeneous, conjugacy-invariant pseudo-length function $l$ on $\F$, we have $l(\alpha\beta\alpha^{-1}\beta^{-1}) = 0$.
\end{conjecture}

In particular, this conjecture implies that $l$ cannot be a length function. Note that it is natural to focus on the element $\alpha\beta\alpha^{-1}\beta^{-1}$ as a group $G$ is abelian if and only if $xyx^{-1}y^{-1} =1\ \forall x, y\in G$, and we were trying to understand whether there are non-abelian examples of groups with length functions with the desired properties.

Several bounds on $l(\alpha\beta\alpha^{-1}\beta^{-1})$ were obtained from the hypothesis, giving bounds that even went below $1$. However, these methods appeared to stagnate with the best bound obtained a little above $0.9$.

Using computer-assistance, we obtained and posted
 a \href{https://github.com/siddhartha-gadgil/Superficial/wiki/A-commutator-bound}{human readable} proof showing $l(\alpha\beta\alpha^{-1}\beta^{-1}) \leq 0.82$.
 An extract of this proof is below\footnote{\url{https://github.com/siddhartha-gadgil/Superficial/wiki/A-commutator-bound} for the full proof as originally posted.}. Note that we have used somewhat different notation here -- the generators are $a$ and $b$ and their inverses are denoted $\bar{a}$ and $\bar{b}$. We remark that a fully expanded proof actually had over 2000 lines, but avoiding duplication gave the posted $126$ lines.

{\tiny\begin{itemize}[leftmargin=*]
	\item $|\bar{a}| \leq 1.0$
	\item $|\bar{b}\bar{a}b| \leq 1.0$ using $|\bar{a}| \leq 1.0$
	\item $|\bar{b}| \leq 1.0$
	\item $|a\bar{b}\bar{a}| \leq 1.0$ using $|\bar{b}| \leq 1.0$
	\item $|\bar{a}\bar{b}ab\bar{a}\bar{b}| \leq 2.0$ using $|\bar{a}\bar{b}a| \leq 1.0$ and $|b\bar{a}\bar{b}| \leq 1.0$
	\item ... (119 lines)
	\item $|ab\bar{a}\bar{b}ab\bar{a}\bar{b}ab\bar{a}\bar{b}ab\bar{a}\bar{b}ab\bar{a}\bar{b}ab\bar{a}\bar{b}ab\bar{a}\bar{b}ab\bar{a}\bar{b}ab\bar{a}\bar{b}ab\bar{a}\bar{b}ab\bar{a}\bar{b}ab\bar{a}\bar{b}ab\bar{a}\bar{b}ab\bar{a}\bar{b}ab\bar{a}\bar{b}ab\bar{a}\bar{b}ab\bar{a}\bar{b}| \leq 13.859649122807017$ \\ using $|ab\bar{a}| \leq 1.0$ and \\ $|\bar{b}ab\bar{a}\bar{b}ab\bar{a}\bar{b}ab\bar{a}\bar{b}ab\bar{a}\bar{b}ab\bar{a}\bar{b}ab\bar{a}\bar{b}ab\bar{a}\bar{b}ab\bar{a}\bar{b}ab\bar{a}\bar{b}ab\bar{a}\bar{b}ab\bar{a}\bar{b}ab\bar{a}\bar{b}ab\bar{a}\bar{b}ab\bar{a}\bar{b}ab\bar{a}\bar{b}ab\bar{a}\bar{b}| \leq 12.859649122807017$
	\item $|ab\bar{a}\bar{b}| \leq 0.8152734778121775$ using \\ $|ab\bar{a}\bar{b}ab\bar{a}\bar{b}ab\bar{a}\bar{b}ab\bar{a}\bar{b}ab\bar{a}\bar{b}ab\bar{a}\bar{b}ab\bar{a}\bar{b}ab\bar{a}\bar{b}ab\bar{a}\bar{b}ab\bar{a}\bar{b}ab\bar{a}\bar{b}ab\bar{a}\bar{b}ab\bar{a}\bar{b}ab\bar{a}\bar{b}ab\bar{a}\bar{b}ab\bar{a}\bar{b}ab\bar{a}\bar{b}| \leq 13.859649122807017$ by \\ taking 17th power.
\end{itemize}}

This proof was studied, understood and generalized by Pace Nielsen, who called the method the \emph{internal repetition trick}. After several improvements due to Nielsen and Tobias Fritz, this was abstracted by Terrence Tao as the following lemma. Note that this holds for all conjugacy-invariant, homogeneous pseudo-lengths $l$ on all groups $G$.
\begin{lemma}\label{rainbow}
	Let $x$, $y$, $z$, $w$ in $G$ be such that $x$ is conjugate
to both $wy$ and $zw^{-1}$, i.e., there exist elements $s, t\in G$ such that $x=swys^{-1}=tzw^{-1}t^{-1}$. Then one has
$$l(x)\leq\frac{l(y) + l(z)}{2}.$$

\end{lemma}

Fritz used this to obtain the key lemma.

\begin{lemma}\label{fritz}
	Let $f(m,k)=l(x^m (xyx^{-1}y^{-1})^k)$. Then $$f(m,k)\leq \frac{f(m-1,k)+f(m+1,k-1)}{2}.$$
\end{lemma}

We apply this lemma to $x=\alpha$, $y=\beta$. An argument based on probability theory, due to Tao, showed that $l(\alpha\beta\alpha^{-1}\beta^{-1})=0$. This in particular answered Question~\ref{qn} (the main result proved in~\cite{polymath} is actually stronger than the Theorem~\ref{main}).
\begin{theorem}[see ~\cite{polymath}]\label{main}
  For every homogeneous, conjugacy-invariant pseudo-length function $l:\F\to [0, \infty)$, we have $l(\alpha\beta\alpha^{-1}\beta^{-1}) =0$. In particular $l$ is not a length function.
\end{theorem}

We mention some of the ingredients in the proof of Theorem~\ref{main} using Lemma~\ref{fritz} with $x=\alpha$ and $y=\beta$. Consider a random walk on points $(m, k)\in\Z^2$ where we move to $(m- 1, k)$ (i.e., one step to the left) with probability $1/2$ and to $(m+ 1, k-1)$ (i.e, diagonally down and right) with probability $1/2$. Then Lemma~\ref{fritz} says that $f(m, n)$ is at most the \emph{average} value of $f$ after one step, and hence inductively after $n$ steps for $n\in \N$. Also observe that we move on an average $1/2$ a step downwards (the left and right movements cancel on an average). Hence if we start at a point $(0, n)$,
(where $n\in \N$) the distribution after $2n$ steps is centered around the origin, and $f(0, n)$ is bounded by the average value of $f$ on this distribution. This average in turn can be bounded using the Chebyshev inequality (as in the proof of the law of large numbers) together with the observation that $f(k, l)\leq m + 2k$ if $l$ is normalized (the latter follows by a straightforward inductive use of the triangle inequality and conjugacy invariance). The bound thus obtained is of the form $f(0, n) \leq C\sqrt{n}$ for some constant $C \in \R$. Finally, as homogeneity gives $l(\alpha\beta\alpha^{-1}\beta^{-1}) \leq f(0, n)/n$, taking a limit as $n\to\infty$ gives $l(\alpha\beta\alpha^{-1}\beta^{-1}) = 0$.

\section{The Algorithms}

Our proof was generated by a mixture of algorithms and expert guidance (with some arbitrary choices). More precisely, given certain \emph{auxiliary choices}, a \emph{deterministic} algorithm gave upper bounds $L(g)$ such that $l(g)\leq L(g)$ for all normalized, homogeneous, conjugacy-invariant pseudo-length functions $l: \F\to\R$ and for all $g\in\F$. 

The auxiliary choices were a finite sequence of pairs $(g_i, n_i)$, with $g_i\in\F$ and $n_i\in\N$. We used homogeneity \emph{only} for these pairs. 
Thus, our algorithm actually gives an upper bound for all functions $l: \F \to \R$ such that 
\begin{itemize}
	\item $l$ is a normalized, conjugacy-invariant length function on $\F$.
	\item $l(g_i) \leq l(g_i^{n_i})/n_i$ for all pairs $(g_i, n_i)$.
\end{itemize}

We shall call such pairs $(g_i, n_i)$ \emph{homogeneity pairs} and a sequence of homogeneity pairs a \emph{homogeneity pair sequence}. Explicit choices for such pairs that give a proof similar to the posted one are given in~\ref{S:Results}, along with links to a script to replicate this (which runs in under 10 seconds on a moderately powerful laptop/desktop). We discuss in~\ref{S:Expertless} how plausible it is to have arrived at such choices through general principles, without expert guidance.

We used a deterministic algorithm (depending on a homogeneity pair sequence) to obtain upper bounds $L(g)$ with $l(g) \leq L(g)$ for all lengths as above and for all $g\in \F$.
Using this, we computed the bound 
$$l(\alpha\beta\alpha^{-1}\beta^{-1}) \leq \min\{\frac{L((\alpha\beta\alpha^{-1}\beta^{-1})^n)}{n}) : 1 \leq n \leq 20\}.$$
This (after keeping track of inequalities and rendering in human readable form) was the posted proof.

All pseudo-lengths we consider henceforth will be assumed to be normalized and conjugacy-invariant (but not necessarily homogeneous).

\subsection{Maximal homogeneous pseudo-lengths}

It is convenient to reformulate our main problem using a standard construction. Namely, we define a function $l_h: \F \to \R$ by defining, for $g \in \F$, $l_h(g)$ to be
$$\max\{ l(g) : \textrm{ $l$ normalized, homogeneous, conjugacy-invariant pseudo-length}\}.$$
 It is well known that this is well-defined and gives the maximal normalized, homogeneous, conjugacy-invariant pseudo-length on $\F$. Thus, our main problem is equivalent to finding upper bounds for $l_h(g)$, in particular for $l_h(\alpha\beta\alpha^{-1}\beta^{-1})$.

\subsection{Bounding conjugacy-invariant pseudo-lengths}	

We now make analogous constructions dropping the homogeneity condition. Let $\mathcal{L}_c$ be the set of all normalized,  conjugacy-invariant pseudo-length functions on $\F$. We define, for $g\in\F$, 
$$l_c(g) = \max\{ l(g) : l\in\mathcal{L}_c \}.$$
This is well-defined and gives the maximal normalized, conjugacy-invariant pseudo-length on $\F$ (i.e., the maximal element of $\mathcal{L}_c$). Further, clearly $l_h(g) \leq l_c(g)$, so upper bounds on $l_c(g)$ give ones on $l_h(g)$. 

We describe in Section~\ref{S:Algo} an algorithm to obtain an upper bound $L_c(g)$ for $l_c(g)$. Indeed this bound is sharp, i.e., we have $l_c(g) = L_c(g)$ for all $g\in \F$ (we do not prove or use this, but this fact motivated our approach). Here and henceforth we follow the convention that we use $l$ with subscripts to denote pseudo-lengths we wish to bound (whose definition may be non-constructive) and $L$ with the same subscript to denote algorithmic upper bounds for these lengths.

\subsection{Conjugacy-invariant lengths with elementary bounds}

Next, suppose we are given a finite set  $B$ of pairs $(g_i, x_i)$, $1 \leq i\leq m$, with $g\in \F$ and $x_i\geq 0$, $x_i\in\R$ (we call this a set of \emph{elementary bounds}). We consider a refinement of $l_c$ and a corresponding modified algorithm (our definitions and algorithms do not depend on the order of the pairs $(g_i, x_i)$).  

Namely, for $g\in\F$, we define 
$$l_b(g; B) = \max\{ l(g) : l\in\mathcal{L}_c,\ l(g_i)\leq x_i\ \forall(g_i, x_i)\in B\}.$$
The function $l_b(g; B)$ is a normalized, conjugacy-invariant pseudo-length on $\F$ which is maximal among such lengths that satisfy the additional bounds $l(g_i) \leq x_i$ for all $(g_i, x_i)\in B$. 

Note that $l_c(g) = l_b(g; \emptyset)$ for $g\in\F$. A straightforward modification of the algorithm describing $L_c(g) = l_c(g)$ gives an algorithm giving bounds $L_b(g; B)$ such that $l_b(g; B) \leq L_b(g; B)$ for all $g\in \F$. We remark that the bound given by this algorithm is not optimal.\footnote{Indeed, an optimal algorithm for $l_b(g; B)$ for general finite $B$ gives a solution to the word problem for groups, which is known to be algorithmically undecidable. Namely, given \emph{relations} $r_1 \in \F$, $r_2\in \F$, \dots $r_m \in \F$, let $B$ be the set $\{(r_1, 0), (r_2, 0), \dots, r_n, 0)\}$. Then $l_b(g; B)= 0$ if and only if $g$ is trivial in the group $\langle \alpha, \beta; r_1 = e$, $r_2 = e, \dots r_m =e \rangle$.}

We shall say that the set of elementary bounds is \emph{admissible} if $l_h(g_i)\leq x_i$ for all $i$, $1\leq i\leq m$. Observe that we can algorithmically obtain an admissible set of elementary bounds from a homogeneity pair sequence $(g_i, n_i)$, $1\leq i\leq m$, by setting $x_i = \frac{l_c(g_i^{n_i})}{n_i} =  \frac{L_c(g_i^{n_i})}{n_i}$ as $l_h(x_i) = \frac{l_h(g_i^{n_i})}{n_i} \leq \frac{l_c(g_i^{n_i})}{n_i} = x_i$.

Note that if a set of elementary bounds $B$ is admissible, then $l_h(g)\leq l(g; B)$ for all $g\in \F$. Hence $L_b(g; B)$ gives an upper bound for $l_h$. We use such a bound, but with the process of obtaining elementary bounds from a homogeneity pair sequences a refinement of setting $x_i = \frac{l_c(g_i^{n_i})}{n_i}$ (and depending on the order of the pairs).

\subsection{Bounds with homogeneity pair sequences}\label{S:Homogeneity}

In this section we describe algorithms depending on a \emph{homogeneity pair sequence} in terms of algorithms depending on \emph{elementary bounds}, essentially by deducing elementary bounds using homogeneity. The algorithms depending on elementary bounds are described in~\ref{S:Algo}, which the reader may prefer to read first. In~\ref{S:PairAlgo} we describe how to modify the algorithms of~\ref{S:Algo} along the lines described below.

Assume that we are given a homogeneity pair sequence, i.e., a finite sequence of pairs $(g_i, n_i)$, $1 \leq i \leq m$. We define inductively in $j$ (simultaneously) 
\begin{itemize}
	\item an elementary bound $(g_j, x_j)$ (with  the element $g_j$ from the given homogeneity pair sequences), 
	\item a length function $l_j:\F \to\R$, such that $l_h(g)\leq l_j(g)$ for all $g\in\F$, and
	\item An \emph{algorithmically defined} length function  $L_j:\F \to\R$, $0\leq i \leq m$,  such that $l_j(g)\leq L_j(g)$ for all $g\in\F$.
\end{itemize}

First, let $l_0(g) = l_c(g)$. Let $L_0(g) = L_c(g)$, which we recall can be computed algorithmically (as described in~\ref{S:Algo}). 

Next, let $x_1 = \frac{l_0(g_1^{n_1})}{n_1} = \frac{L_c(g_1^{n_1})}{n_1}$ and define $l_1(g) = l_b(g; \{(g_1, x_1)\})$.  By homogeneity, $l_h(g_1)\leq x_1$, so by maximality of $l_b(g; \{(g_1, x_1)\})$, $l_h(g)\leq l_1(g)$ for all $g\in\F$.

Recall that we have an algorithm (described in~\ref{S:Algo}) giving (for $g\in\F$) an upper bound  $L_b(g; \{(g_1, x_1)\})$ for $l_b(g; \{(g_1, x_1)\})$. Define $L_1(g) = L_b(g; \{(g_1, x_1)\})$ for $g \in \F$.

Continuing in this fashion define 
\begin{itemize}
	\item $x_2 = \frac{L_1(g_2^n)}{n_2}$ (which can be algorithmically computed),
	\item $l_2(g) = l_b(g; \{(g_1, x_1), (g_2, x_2)\})$, and
	\item $L_2 = L_b(g; \{(g_1, x_1), (g_2, x_2)\})$ (which is algorithmic).
\end{itemize}

As before, we have the bounds $l_h(g) \leq l_2(g)$ for all $g\in\F$. 

Inductively, given $k<m$, $x_i\in \F$ for $1\leq i \leq k$, a function $l_k: \F \to \R$, and an algorithmically defined function $L_k: \F \to \R$, define 

\begin{itemize}
	\item $x_{k+1} = \frac{L_k(g_{k+1}^{n_{k+1}})}{n_{k+1}}$ (which can be algorithmically computed),
	\item $l_{k+1}(g) := l_b(g; \{(g_1, x_1), (g_2, x_2), \dots, (g_{k+1}, n_{k+1})\})$, and 
	\item $L_{k+1}(g) := L_b(g; \{(g_1, x_1), (g_2, x_2), \dots, (g_{k+1}, n_{k+1})\})$
\end{itemize}
 The function $L(g) := L_m(g) = L_b(g; \{(g_1, x_1), (g_2, x_2), \dots, (g_{n}, n_{m})\})$ is the desired upper bound for $l_h$.

\subsection{Algorithm for conjugacy-invariant pseudo-lengths}\label{S:Algo}

We now describe the algorithms giving $L_c(g)$ and $L_b(g; B)$, i.e. giving upper bounds for $l_c(g)$ and $l_b(g; B)$. Recall that $l_c(g) = l_b(g ; \emptyset)$.

Let $l$ be a normalized, conjugacy-invariant pseudo-length $l(g)$, which may also be assumed to satisfy a finite number of elementary bounds.
We describe an upper bound  for $l(g)$ for a word $g=\xi_1\xi_2\dots \xi_n$ recursively in the length $n$ of the word. The key ingredient is the following lemma bounding $l(g)$ in terms of bounds on shorter words.
\begin{lemma}\label{recbound}
	Let $g=\xi_1\xi_2\dots \xi_n$ with $n > 1$.
	\begin{enumerate}[label=(\alph*)]
		\item\label{plus1} $l(g) \leq 1 + l(\xi_2\xi_3\dots \xi_n)$
		\item\label{pair} If $\xi_k = \xi_1^{-1}$, then $l(g)\leq l(\xi_2\xi_3\dots \xi_{k-1}) + l(\xi_{k+1}\xi_{k+2}\dots \xi_n)$
	\end{enumerate}
\end{lemma}
\begin{proof}
	To see~\ref{plus1}, observe that
	\begin{align*}
			l(g) &= l(\xi_1\xi_2\dots \xi_n) & \\
			&\leq l(\xi_1) + l(\xi_2\xi_3\dots \xi_n), & (\textrm{by the triangle inequality}) \\
			&\leq 1 + l(\xi_2\xi_3\dots \xi_n), & (\textrm{as $l$ is normalized})
	\end{align*}
as claimed.

	Next, suppose $\xi_k = \xi_1^{-1}$. By the triangle inequality,
	\begin{align}
	l(g) &\leq l(\xi_1\xi_2\xi_3\dots \xi_{k-1}\xi_k) + l(\xi_{k+1}\xi_{k+2}\dots \xi_n) \\
			&= l(\xi_1(\xi_2\xi_3\dots \xi_{k-1})\xi_1^{-1}) + l(\xi_{k+1}\xi_{k+2}\dots \xi_n).\label{eq:triang}
	\end{align}
	Further, by conjugacy invariance of $l$,
\begin{equation}\label{eq:conj}
	l(\xi_1(\xi_2\xi_3\dots \xi_{k-1})\xi_1^{-1}) = l(\xi_2\xi_3\dots \xi_{k-1}).
\end{equation}

	Substituting~\eqref{eq:conj} in~\eqref{eq:triang}, we get
	 $$l(g)\leq l(\xi_2\xi_3\dots \xi_{k-1}) + l(\xi_{k+1}\xi_{k+2}\dots \xi_n),$$ showing~\ref{pair}.
\end{proof}

The algorithms are based on Lemma~\ref{recbound}.
Elementary bounds $l(g_i)\leq x_i$, $1\leq i\leq m$ are specified by a map $L_0: D \to\R$ with $D \subset \F = \{g_1, \dots, g_m\}$ and $L_0(g_i) = x_i$. If we have no such bounds, i.e. we are computing $l_c$, we  initially take $D=\emptyset$ and $L_0$ as the empty map (however the map is updated to avoid repeating computations). 

For $g\in \F$, the recursive algorithm shown in Figure~\ref{algo} gives a bound $L(g)$ so that $l(g) \leq L(g)$ for any normalized, conjugacy-invariant pseudo-length $l$ on $\F$. 
 We describe this using sets in mathematical language, but this can be readily translated to code using, for instance, list comprehensions.

\begin{figure}
	\caption{Algorithm for bounding lengths}\label{algo}
\medskip
\begin{mdframed}
	\begin{itemize}
		\item For $g \in \F$, compute $L(g)$ by
		\begin{itemize}
			\item If $g = e$ is the empty word, \textbf{define} $L(g) := 0$.
			\item If $g=\xi_1$ has exactly one letter, \textbf{define} $L(g) := 1$.
			\item If $d\in D$, \textbf{define} $L(g) = L_0(g)$.
			\item If $g = \xi_1\xi_2\dots \xi_n$ has at least two letters (and $g\notin D$):
			\begin{itemize}
				\item let $\lambda_0 = 1 + L(\xi_2\xi_3\dots \xi_n)$ (computed recursively).
				\item for $2\leq k \leq n$, define $$\lambda_k = L(\xi_2\xi_3\dots \xi_{k-1}) + L(\xi_{k+1}\xi_{k+2}\dots \xi_n).$$
				\item let $\Lambda$ be the set
				$$\Lambda =  \{\lambda_k : 2 \leq k\leq n, \xi_k = \xi_1^{-1}\}.$$
				\item let $x = \min(\{\lambda_0\}\cup \Lambda)$.
				\item let $D := D\cup\{g\}$ and extend $L_0$ by defining $L_0(g)= x$.
				\item \textbf{define} $L(g) = x$.
			\end{itemize}
	\end{itemize}
\end{itemize}
\end{mdframed}
\end{figure}

Furthermore, we can keep track of a labelled rooted tree of inequalities used to compute $L(g)$, and hence bound $l(g)$. We give a schematic condensed example of such a tree in Figure~\ref{tree}. Essentially such a tree (actually a corresponding Algebraic Data Type) was used to generate the human readable proof.

\begin{figure}
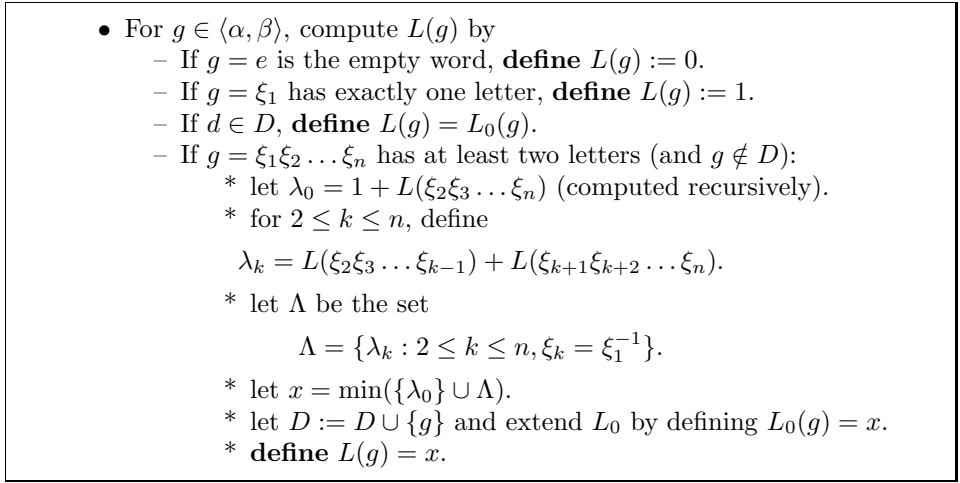

	\caption{Proof tree of $l(\alpha\beta\alpha^{-1}\beta^{-1}) \leq 1$ in a YAML-like format}\label{tree}
	\medskip
\begin{mdframed}
\begin{itemize}
	\item bound: $l(\alpha\beta\alpha^{-1}\beta^{-1}) \leq 1$
	\item proof: triangle-inequality
	\begin{itemize}
		\item first-inequality:
		\begin{itemize}
			\item bound: $l(\alpha)\leq 1$
			\item proof: length-is-normalized
		\end{itemize}
		\item second-inequality
		\begin{itemize}
			\item bound: $l(\beta\alpha^{-1}\beta^{-1})\leq 1$
			\item proof: conjugacy-invariance
			\begin{itemize}
				\item conjugated-by: $\beta$
				\item base-inequality:
				\begin{itemize}
					\item bound: $l(\alpha^{-1}) \leq 1$.
					\item proof: length-is-normalized.
				\end{itemize}
			\end{itemize}
		\end{itemize}
	\end{itemize}
\end{itemize}
\end{mdframed}
\end{figure}

Observe that the function $L$ is integer valued, and can hence be computed exactly. On using homogeneity we obtain rational bounds. These were stored as double precision real numbers -- we switched to arbitrary precision rational numbers at one stage, but switched back for performance reasons during the recursive computations (as these were involved in  a large search). Note however that it is easy (and efficient) to map a proof tree using doubles to one using arbitrary precision rational numbers to ensure that there is no error in rounding off. Indeed, as we discuss in~\ref{S:Results}, we have subsequently implemented the mapping of proofs to exact rational bounds.

\subsection{Bounding with a homogeneity pair sequence}\label{S:PairAlgo}

We now describe how to modify the above algorithm given a homogeneity pair sequence. This is following the approach of~\ref{S:Homogeneity}, but with one minor difference as mentioned in Remark~\ref{R:Memo}.

Suppose now that we are given a homogeneity pair sequence, i.e., a finite sequence of pairs $(g_i, n_i$), $1\leq i \leq m$. We initialize $L_0$ to be the empty map, and compute $L(g_1^{n_1})$ using the main algorithm. We let $x_1 = L(g_1^{n_1})/n_1$ and update the map $L_0$ by setting $L_0(g_1) = x_1$.

Next we use the main algorithm to compute $L(g_2^{n_2})$, but with $L_0$ the map obtained at the end of the previous computation and after setting $L_0(g_1) = x_1$. Again let $x_2 = L(g_2^{n_2})/n_2$ and update the map $L_0$ by setting $L_0(g_2) = x_2$. We proceed in this fashion to obtain the numbers $x_1$, $x_2$, \dots, $x_n$, and an \emph{updated} map $L_0$.

Finally, for an elements $g\in\F$, we define the function $L(g)$ as the result of using the algorithm with the map $L_0$ obtained at the end of the above sequence of computations and updates.

\begin{remark}\label{R:Memo}
	While the above algorithm is very similar to that described in~\ref{S:Homogeneity}, it gives in general slightly worse bounds. This is because, for a fixed $g_0\in\F$, if $L(g_0)$ is computed for a word while computing, for example $x_1$ (which happens if $g_0$ is a subword of $g_1^{n_1}$), we do not recompute $L(g_0)$ when computing, for example $x_2$. However we may get a smaller value (i.e., better bound) if we recomputed $L(g_0$) as we have the additional elementary bound $L(x_1) \leq g_1$ (we get an improved bound if $g_1$ is a subword of $g_0$ and $L(g_1^{n_1}) < n_1L(g_1)$). It is easy to avoid this by setting the map $L_0$ when computing $x_i$ to be just the earlier elementary bounds, i.e. set $D=\{g_j: j < i\}$ and $L_0(g_j) = x_j$. However, this comes at a cost in efficiency due to computations being repeated. 
\end{remark}

\subsection{Choices and results}\label{S:Results}

In generating the proofs, we used the family of words of the form $\gamma_k = \alpha(\alpha\beta\alpha^{-1}\beta^{-1})^k$, with this family chosen based on expert knowledge. From these, we constructed a homogeneity sequence, depending on certain choices. 

Namely, our homogeneity pair sequence was of the following form (with explicit choices stated, which we have used in a script as mentioned below):
\begin{itemize}
	\item We choose and fix $N \geq 1$ (we take $N=20$).
	\item Choose and fix a few values of $k$ (chosen with some experimentation), say $k_1$, $k_2$, \dots, $k_m$ (we take $m=3$ with $k_1=1$, $k_2=2$ and $k_3= 6$).
	\item We get a homogeneity pair sequence taking each element $\gamma_{k_i}$ with each exponent between $1$ and $N$, namely
	$$(\gamma_{k_1}, 1), \dots, (\gamma_{k_1}, N), (\gamma_{k_2}, 1), \dots, (\gamma_{k_2}, N), \dots, (\gamma_{k_m}, 1), \dots, (\gamma_{k_m}, N).$$
	\item The homogeneity pair sequence we use is the above sequence \emph{followed by} the sequence $$(\alpha\beta\alpha^{-1}\beta^{-1}, 1), (\alpha\beta\alpha^{-1}\beta^{-1}, 2) \dots, (\alpha\beta\alpha^{-1}\beta^{-1}, N)$$ 
\end{itemize}

With the explicit choices $N=20$, $m=3$, $k_1=1$, $k_2=2$ and $k_3= 6$, we get the bound $$l_h(\alpha\beta\alpha^{-1}\beta^{-1})\leq 0.8098765432098762$$
and a corresponding human readable proof. 

Furthermore, we can map proofs to  arbitrary precision rational bounds to avoid rounding errors. Mapping the above proof gives the bound  (with no rounding-off error)
$$l_h(\alpha\beta\alpha^{-1}\beta^{-1})\leq 328/405.$$

We have created an \emph{executable jar} file to replicate generating this proof (as well as mapping to a proof with arbitrary precision rational bounds). This is available online at~\url{http://math.iisc.ac.in/~gadgil/PolyProof.html} (with instructions on running it), along with sample output (slightly reformatted). On the systems we used (a desktop and a laptop with Core~i7 processors) this runs in under 10 seconds.\footnote{The full code is in the repository  \href{https://github.com/siddhartha-gadgil/Superficial}{\texttt{https://github.com/siddhartha-gadgil/Superficial}}. The script is generated from this source. The script uses the same algorithms we originally used, but with modifications to be more robust in memory usage and to avoid concurrency (as the concurrency we implemented leads to non-determinacy, and occasionally to race conditions).} The proof generated by this script is a little longer than the one originally posted ($173$ lines instead of $126$) but gives a slightly better bound.

Values for $N$ and the indices $k_i$ were obtained by experimentation, and the bounds are fairly robust when we vary choices.

\begin{remark}
	In generating the script we only used our knowledge that $k=6$ was useful (though not crucial) while generating the original proof, and the only other choices we tried were taking $N = 10$, which also gives a bound below $1$, and also taking $k$'s to be $1$, $2$ ,$3$ and $6$, which only marginally improved the bound. 
\end{remark}

The choice of the family $\gamma_k$ was based on mathematical considerations related to~\cite{gadgil}, and this was the only expert guidance. We next discuss whether finding the proof was plausible using general principles in place of expert knowledge.

\subsection{Auxiliary choices without expert knowledge?}\label{S:Expertless}

As we have seen, the only expert guidance was the ``natural family of group elements $\gamma_k$". We sketch a series of general considerations (some using basic group theory) that could plausibly have led to the same family.

\begin{itemize}
	\item A natural measure of \emph{usefulness} of a homogeneity pair $(g, n)$ is the ratio $\rho(g, n) = \frac{l_c(g)}{l_c(g^n)/n}$, as this is an upper bound on the ratio $l_c(h)/l_b(h; \{(g, n)\})$ for $h\in \F$, i.e., the maximum possible improvement in bounds (a value of $1$ means no gain from using homogeneity). 
	\item Rather than looking for individual useful elements, we look for \emph{families} $\gamma_k$ of useful elements, choosing between families by small scale sampling.
	\item We look for \emph{natural} families in the sense of having a simple description in terms of the group operations. The simplest families in a group are those of the form $a^k$ for fixed $a \in \F$ and the next simplest are those of the form $ab^k$ for fixed $a, b\in\F$. The family we considered is of the second form. (As $b^k a$ is conjugate to $a^{-1}b^k$, families of the form $b^ka$ are equivalent to those we considered.)
	\item We first consider simple families, i.e., with simple words for $a$ and $b$, while using symmetries of the problem to reduce choices. 
\end{itemize}

Here the symmetries are: transposing the generators $\alpha$ and $\beta$, transposing one or both generators with their inverses, and cyclic permutations of words. Further,  one needs to only consider \emph{reduced} words, i.e., those without a cancelling pair. Up to all these symmetries, there are only $6$ words with length $4$, $2$ each with lengths $3$ and $2$ and a single word with length $1$. Even allowing for not all symmetries being exploited (as they do depend on some expert knowledge) and allowing for different choices for the word $a$, the number of families of complexity comparable to the one we considered is modest.

Furthermore, if $b$ is a word of length at most $4$ which is not equivalent to $\alpha\beta\alpha^{-1}\beta^{-1}$, then $l_h(b) = l_c(b)$, which in particular implies that $\rho(ab^k, n) \approx 1$ for large $k$ and $n$. Thus the families with other values of $b$ can be ruled out as not useful with limited experimentation. 

Thus, if one searches through natural families, with simpler ones considered first and symmetries exploited to avoid duplication, and assesses each family rapidly by measuring improvements in relevant bounds after small scale sampling, one is likely to arrive at the family we considered (or an equivalent one) in a reasonable amount of time.

\section{Concluding remarks}\label{S:Conclusions}

If we view applications of the axioms as \emph{moves},  the computer proof helped in identifying \emph{composite moves} that could be applied \emph{recursively} for words in appropriate families. These were abstracted and generalized to give the core lemma. One can hope that in other situations as well computer generated proofs targeting key examples give hints about useful composite moves, and especially those that can be used iteratively.

The principal difficulty in finding computer proofs often lies in choosing the useful moves among those that increase complexity, which in our case are applications of homogeneity. In this work, we primarily based ourselves on mathematical considerations, partly because the auxiliary choices were identified even before we started programming, and partly because of the fast pace\footnote{the proof was posted less than two days after we began writing code, and the main question answered less than a day after that.}. However, we have attempted to justify in Section~\ref{S:Expertless} that it is plausible that experimentation and general considerations could have led to similar results, with the use of one heuristic -- one should search for \emph{natural families} of useful moves. 

We used modest computing resources (including time) and did not use search heuristics. It is thus all the more likely that domain specific expertise could have been replaced by (or combined with) the vast arsenal of well-known heuristics for tree searches, such as Alpha-beta pruning, Markov Chain Monte Carlo and various Machine learning techniques.

Fully automating the search in this and similar problems involves identifying useful candidate families (beyond simply enumerating those with simple descriptions). This is an interesting challenge, and it would be interesting to explore various techniques for this. Computer scientists have developed various techniques to find inputs trigger a bug (such as fuzzing) and techniques to minimize such inputs (such as delta-debugging). Perhaps these 
techniques are also useful for finding a human-understandable proof like the one presented in the paper.

\begin{acknowledgements}
	I thank the referees and the editors for many valuable comments, which have led to the paper being completely rewritten twice and much improved in the process. It is also a pleasure to thank the rest of the PolyMath 14 team for the collaboration of which the work described here is a part. 
	
\end{acknowledgements}

\end{document}